\documentclass[11pt]{article}
\usepackage[T1]{fontenc}
\usepackage[utf8]{inputenc}
\usepackage{lmodern}
\usepackage[lmargin=1.0in,rmargin=1.0in,bottom=1.0in,top=1.0in,twoside=False]{geometry}
\usepackage{fullpage,amssymb,amsmath}
\usepackage{graphicx}
\usepackage{enumerate}

\usepackage{xcolor}
\usepackage{mathtools}
\usepackage{microtype}
\usepackage{amsfonts}
\usepackage[initials]{amsrefs}
\usepackage{comment}
\usepackage[UKenglish]{babel}
\usepackage{mathrsfs}
\usepackage[ruled,vlined]{algorithm2e}
\usepackage{blindtext}

\usepackage{trimspaces}
\usepackage{nccfoots}
\usepackage{setspace}
\usepackage{inconsolata}
\usepackage{libertine}
\usepackage[absolute]{textpos}

\usepackage{enumitem}
\usepackage{todonotes}

\usepackage{wasysym}

\usepackage{longtable}

\definecolor{blue}{rgb}{0.1,0.2,0.5}
\definecolor{brown}{rgb}{0.6,0.6,0.2}
\usepackage[ocgcolorlinks, linkcolor={blue}, citecolor={brown}]{hyperref}

\usepackage[amsmath,thmmarks,hyperref]{ntheorem}
\usepackage{cleveref}

\crefformat{page}{#2page~#1#3}%
\Crefformat{page}{#2Page~#1#3}%
\crefformat{equation}{#2(#1)#3}%
\Crefformat{equation}{#2(#1)#3}%
\crefformat{figure}{#2Figure~#1#3}%
\Crefformat{figure}{#2Figure~#1#3}%
\crefformat{section}{#2Section~#1#3}
\Crefformat{section}{#2Section~#1#3}
\crefformat{chapter}{#2Chapter~#1#3}
\Crefformat{chapter}{#2Chapter~#1#3}
\crefformat{chapter*}{#2Chapter~#1#3}
\Crefformat{chapter*}{#2Chapter~#1#3}
\crefformat{part}{#2Part~#1#3}
\Crefformat{part}{#2Part~#1#3}
\crefformat{enumi}{#2(#1)#3}
\Crefformat{enumi}{#2(#1)#3}

\usepackage{enumerate}

\usepackage{latexsym}

\theoremnumbering{arabic}
\theoremstyle{plain}
\theoremsymbol{}
\theorembodyfont{\itshape}
\theoremheaderfont{\normalfont\bfseries}
\theoremseparator{.}

\newtheorem{theorem}{Theorem}
\crefformat{theorem}{#2Theorem~#1#3}
\Crefformat{theorem}{#2Theorem~#1#3}

\newcommand{\newtheoremwithcrefformat}[2]{%
  \newtheorem{#1}[theorem]{#2}%
  \crefformat{#1}{##2\MakeUppercase#1~##1##3}%
  \Crefformat{#1}{##2\MakeUppercase#1~##1##3}%
}
\newcommand{\newseptheoremwithcrefformat}[2]{%
  \newtheorem{#1}{#2}%
  \crefformat{#1}{##2\MakeUppercase#1~##1##3}%
  \Crefformat{#1}{##2\MakeUppercase#1~##1##3}%
}

\newtheoremwithcrefformat{lemma}{Lemma}
\newtheoremwithcrefformat{proposition}{Proposition}
\newtheoremwithcrefformat{observation}{Observation}
\newtheoremwithcrefformat{conjecture}{Conjecture}
\newtheoremwithcrefformat{corollary}{Corollary}
\newseptheoremwithcrefformat{claim}{Assertion}
\theorembodyfont{\upshape}
\newtheoremwithcrefformat{example}{Example}
\newtheoremwithcrefformat{remark}{Remark}
\newtheoremwithcrefformat{definition}{Definition}

\theoremstyle{nonumberplain}
\theoremheaderfont{\scshape}
\theorembodyfont{\normalfont}
\theoremsymbol{\ensuremath{\square}}
\newtheorem{proof}{Proof}

\theoremsymbol{\ensuremath{\lrcorner}}

\def\cqedsymbol{\ifmmode$\lrcorner$\else{\unskip\nobreak\hfil
\penalty50\hskip1em\null\nobreak\hfil$\lrcorner$
\parfillskip=0pt\finalhyphendemerits=0\endgraf}\fi} 

\newcommand{\cqed}{\renewcommand{\qed}{\cqedsymbol}}

\newcommand{\fit}{binding}

\newcommand{\N}{\mathbb{N}}
\newcommand{\R}{\mathbb{R}}
\newcommand{\eps}{\varepsilon}

\renewcommand{\P}{\mathbf{P}}
\newcommand{\A}{\mathbf{A}}
\newcommand{\Str}{\Sigma}
\newcommand{\Tct}{\Theta}
\newcommand{\Tctf}{\Tct^{\A}}
\newcommand{\Tctm}{\Tct^{\P}}

\newcommand{\Pb}{\mathbb{P}}
\newcommand{\Exp}{\mathbb{E}}
\newcommand{\Var}{\mathrm{Var}}
\newcommand{\Cov}{\mathrm{Cov}}

\newcommand{\wt}[2]{W(#1,#2)}
\newcommand{\wtime}[2]{T\langle #1,#2\rangle}

\renewcommand{\leq}{\leqslant}
\renewcommand{\geq}{\geqslant}

\begin{document}
\title{\textbf{On the effect of symmetry requirement for rendezvous on the complete graph}\thanks{This work is a part of project TOTAL (Mi. Pilipczuk) that have received funding from the European Research Council (ERC) under the European Union's Horizon 2020 research and innovation programme (grant agreement No.~677651). It is also partially funded by the French ANR projects ANR-16-CE40-0023 (DESCARTES) and ANR-17-CE40-0015 (DISTANCIA).}}

\author{
Marthe Bonamy\thanks{
 CNRS, LaBRI, Université de Bordeaux, France, \texttt{marthe.bonamy@u-bordeaux.fr}.
}
\and
Micha\l{}~Pilipczuk\thanks{
  Institute of Informatics, University of Warsaw, Poland, \texttt{michal.pilipczuk@mimuw.edu.pl}.
}
\and
Jean-S\'ebastien Sereni\thanks{
  Centre national de la recherche scientifique, CSTB (ICube), Strasbourg, France, \texttt{sereni@kam.mff.cuni.cz}.
}
}

\begin{titlepage}
\def\thepage{}
\thispagestyle{empty}
\maketitle

\begin{textblock}{20}(0, 12.5)
\includegraphics[width=40px]{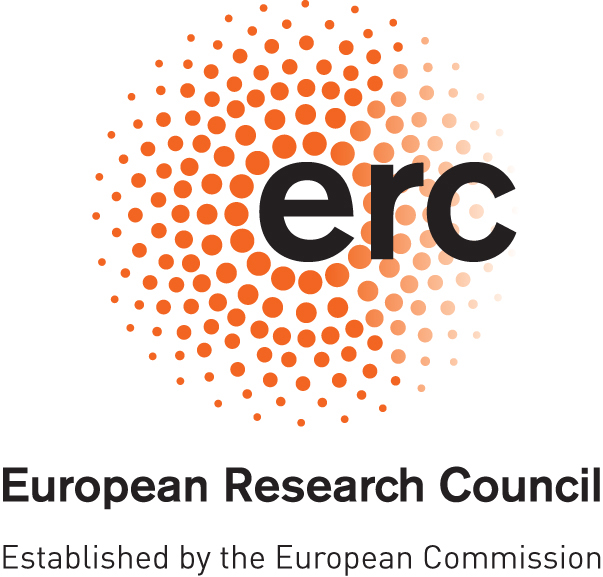}%
\end{textblock}
\begin{textblock}{20}(-0.25, 12.9)
\includegraphics[width=60px]{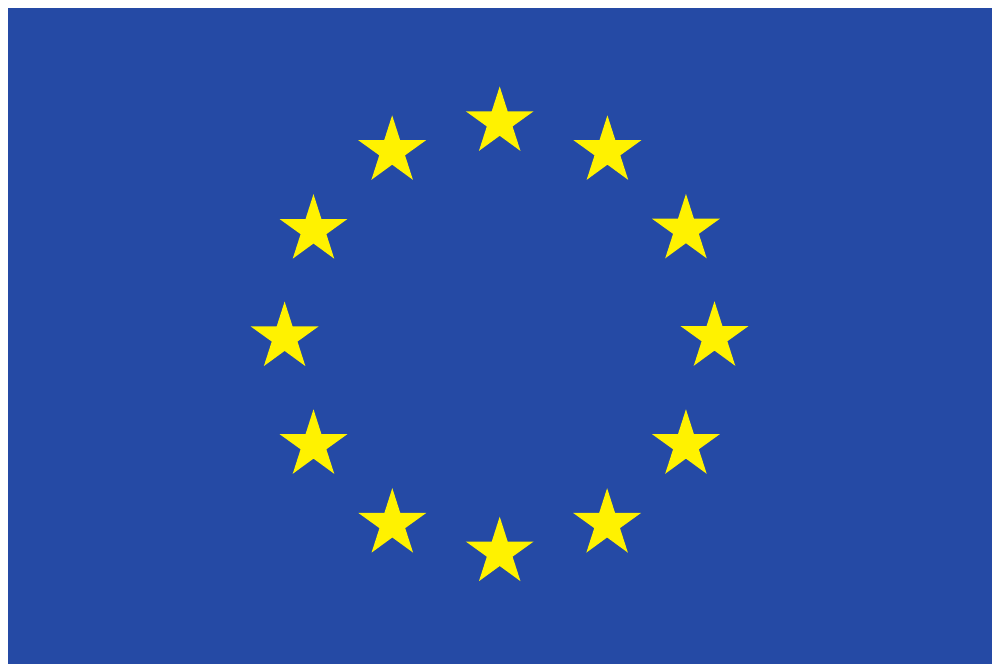}%
\end{textblock}

\begin{abstract}
We consider a classic rendezvous game where two players try to meet each other on a set of $n$ locations.
    In each round, every player visits one of the locations and the game finishes when the players meet
    at the same location. The goal is to devise strategies for both players that minimize the expected
    waiting time till the rendezvous.

In the asymmetric case, when the strategies of the players may differ, it is known that the optimum
    expected waiting time of $\frac{n+1}{2}$ is achieved by the \emph{wait-for-mommy} pair of strategies,
    where one of the players stays at one location for $n$ rounds, while the other player searches
    through all the $n$ locations in a random order. However, if we insist that the players are symmetric
    --- they are expected to follow the same strategy --- then the best known strategy, proposed by
    Anderson and Weber~\cite{AndersonW90}, achieves an asymptotic expected waiting time of $0.829 n$.

We show that the symmetry requirement indeed implies that the expected waiting time needs to be
    asymptotically larger than in the asymmetric case. Precisely, we prove that for every $n\geq 2$, if
    the players need to employ the same strategy, then the expected waiting time is at least
    $\frac{n+1}{2}+\eps n$, where $\eps=2^{-36}$.
\end{abstract}

\end{titlepage}

\section{Introduction}\label{sec:intro}

Rendezvous search questions fall within the long-established field of search games:
instead of having a player searching for an otherwise indifferent treasure, there are
now two players that want to meet as quickly as possible. This very natural problem lends
itself to a number of very different, more or less formalised settings. It was first specified as an optimisation
problem in~1976 by Alpern at the end of a talk (see~\cite{Alpern02}), in two different settings:
the \emph{astronaut problem} and the seemingly simpler \emph{telephone problem}.
In the former problem, two players are on a sphere, each with a given unit walking speed
and no common orientation in space, and they want to minimise their expected meeting time.
The telephone problem has since been rephrased as
a \emph{rendezvous game on discrete locations}, as follows.
Two players wish to meet on a set of $n$ locations and they proceed in
rounds. In each round, every player visits a location of her choice. The game
finishes when both players meet at the same location. The goal of the players
is to minimize the expected waiting time till a meet-up, also called a
\emph{rendezvous}.
This formulation permits to easily impose extra constraints on how the players can move from one location
to another by using different underlying space topologies (which here are graphs: in the original
telephone problem, the underlying graph is the complete graph on~$n$ vertices). The most studied cases
are when the graph is either complete or a path.

Let us point out that, originally, no difference between the two players is assumed here, so that they
must use the same strategy: this is called the \emph{symmetric case}. It implies a level of randomness,
as otherwise the players may well never meet. The case where the players are allowed to use
different strategies, called \emph{asymmetric}, was introduced in~1995~\cite{Alpern95}.

As pointed out earlier, for instance by Alpern~\cite{Alpern02}, such a natural question can be raised
in a number of different contexts, such as that of migrating animals.
There is a rich research literature on
rendezvous games and its many variants, \emph{e.g.} with more players~\cite{AlpernL02},
different rules of the game (for instance seeking to minimise the second meeting time~\cite{Weber12b}) or
other topologies of the search space (including when the players know where they start from, that is,
when they have a common labelling of the graph~\cite{Alpern02b}). We invite an interested reader to the survey of
Alpern~\cite{Alpern02} for a broader and formal introduction.

Coming back to the rendezvous game on the complete graph, the asymmetric case
was solved by Anderson and Weber~\cite{AndersonW90}, using what is coined
the \emph{wait-for-mommy strategy}: one of the players stays for
$n$ rounds in one location, while the other player searches through all the $n$
locations in a random order. Then the expected waiting time is equal to
$\frac{n+1}{2}$, and it is known that this value is optimum~\cite{AndersonW90}: every pair of
strategies for the players yields expected waiting time not lower than
$\frac{n+1}{2}$.

Apart from proving the aforementioned lower bound of $\frac{n+1}{2}$ in the asymmetric case, Anderson and
Weber~\cite{AndersonW90} also studied the symmetric variant of the problem, where the two players are
required to use the same strategy.  While always visiting a random location gives an expected waiting
time of $n$, Anderson and Weber proposed a more clever symmetric strategy that achieves an asymptotic
expected waiting time slightly smaller than $0.829 n$, which we explain next.

The Anderson-Weber strategy works as follows. Let $\theta\in [0,1]$ be a parameter, to be fixed later.
The players divide the game into groups of $n-1$ consecutive rounds. At the beginning of each group of
rounds, each player randomly decides her behavior during these rounds: with probability $\theta$ she will
stay in a random location for all $n-1$ rounds, and with probability $1-\theta$ she will visit $n-1$
locations chosen at random, and in a random order. Thus, intuitively, the Anderson-Weber strategy tries
to break the symmetry by randomly assigning to each player either the role of the baby (who is passive),
or the role of the mommy (who is active). However, there is a significant probability that both players
get the same role, which results in an expected waiting time significantly higher than $\frac{n+1}{2}$.
Indeed, while for different $n$, different values of $\theta$ optimize the expected waiting time, with
$n$ tending to infinity one should pick $\theta$ tending to roughly~$0.24749$, which results in an
asymptotic expected waiting time slightly smaller than $0.829 n$.

The Anderson-Weber strategy has been analyzed for small values of $n$. It is known that picking the right
$\theta$ yields an optimum strategy for $n=2$~\cite{AndersonW90} and for $n=3$~\cite{Weber12}, this
latter result being much more difficult to prove. For $n=4$, as proved by Weber~\cite{Weber09}
there is a slightly better strategy outside of the framework of Anderson and Weber.
However, in general it is conjectured that the Anderson-Weber strategy is asymptotically optimum: there
is no strategy for arbitrary $n$ that would yield an asymptotic expected waiting time smaller than
(roughly)~$0.829 n$. However, to the best of our knowledge, no asymptotic lower bound higher than
$\frac{n+1}{2}$, which holds even for the asymmetric variant, was known prior to this work.

\paragraph*{Our contribution.} We prove that for every $n\geq 2$, in the symmetric rendezvous game on $n$
locations the expected waiting time needs to be significantly larger than $\frac{n+1}{2}$. Precisely, if
the players are requested to follow the same strategy, then whatever strategy they choose, the expected
waiting time will be at least~$\frac{n+1}{2}+\eps n$ for $\eps=2^{-36}$. See Theorem~\ref{thm:main} in
Section~\ref{sec:prelims} for a formal statement. While this still leaves a large gap to the best known
upper bound of $0.829 n$, due to Anderson and Weber~\cite{AndersonW90}, this seems to be the first lower
bound for arbitrary $n$ that significantly distinguishes the symmetric case from the asymmetric case,
where $\frac{n+1}{2}$ is the optimum.

The idea behind our proof can be explained as follows. As in other works, \emph{e.g.}~\cite{Weber12}, we
restrict the game to the first $n$ rounds and prove a lower bound already for this simpler game. We
classify deterministic strategies of the players (which we call \emph{tactics}) into those that rather
stay at few locations and those that seek through many locations. Formally, tactics of the first kind
--- the \emph{passive tactics} --- visit at most $n/2$ different locations, while tactics of the
second kind --- the \emph{active tactics} --- visit more than $n/2$ different~locations.

The intuition drawn from the asymmetric case is that the expected waiting time is minimized when one
player plays a passive tactic, while the other plays an active tactic. As now the players need to follow the
same strategy (understood as a probability distribution over tactics), with probability at least
$\frac{1}{2}$ they choose to use tactics of the same kind (activity level). Then it suffices to prove
that when two tactics of the same kind are played against each other, the expected waiting time is
significantly larger than $\frac{n+1}{2}$.

To this end, we show that if same-kind tactics are employed, the probability that no rendezvous happens
at all is bounded from below by a positive constant. This easily implies a better-than-$\frac{n+1}{2}$
lower bound on the expected waiting time. To analyze the probability of no rendezvous, we investigate a
random variable~$X$ that indicates the total number of rendezvous if the game is not stopped when the
players meet for the first time. Then $X$ has mean (roughly) equal to $1$, so to prove that $X=0$ with
significant probability, we show that $X$ is not well concentrated around its mean. This involves
establishing a lower bound on the variance of $X$, which in turn follows from the assumption that the
employed tactics have the same kind.

\section{The model and the problem}\label{sec:prelims}

In this section we formalize the considered rendezvous search game and state the main result in precise
terms. As in previous works, \emph{e.g.}~\cite{Weber12}, we make the game finite by stopping it after $n$
rounds. Precisely, if the players did not meet after $n$ rounds, we stop the game and set $n+1$ as the
obtained time till rendezvous. Note that this may only decrease the expected waiting time as compared to
allowing the players to play indefinitely.

We are given a set of $n$ locations and two players, $A$ and $B$. Each player has her own, private
numbering of locations using numbers from $[n]\coloneqq \{1,\ldots,n\}$. A \emph{tactic} for a player is
a function $\tau\colon [n]\to [n]$, where~$\tau(i)$ is interpreted as the index of the location that the
player intends to visit at round $i$, in her own numbering. A \emph{strategy} for a player is a
probability distribution $\sigma$ over the tactics of this player. Note that the set of possible tactics
is finite, hence we may use the discrete $\sigma$-field where every subset of tactics is measurable.  The
sets of tactics and strategies for the game on $n$ locations are denoted by $\Tct_n$ and $\Str_n$,
respectively.

For two given strategies $\sigma_A$ and $\sigma_B$, the game is played as follows:
\begin{itemize}
    \item Players $A$ and $B$ respectively draw their tactics $\tau_A$ and $\tau_B$ from the strategies
        $\sigma_A$ and $\sigma_B$ at random.
    \item A permutation $\pi\colon [n]\to [n]$ that matches the numberings of locations of $A$ and $B$ is
        drawn uniformly at random. This permutation $\pi$ will be called the \emph{\fit}.
    \item The waiting time till rendezvous is indicated by the random variable
        \[\wtime{\sigma_1}{\sigma_2}=\min \left(\{\,i\colon \pi(\tau_A(i))=\tau_B(i)\,\}\cup \{n+1\}\right).\]
\end{itemize}
Then the question of minimizing the waiting time till rendezvous for symmetric players corresponds to the problem of minimizing the expected value of $\wtime{\sigma_A}{\sigma_B}$ over the strategies $\sigma_A$ and $\sigma_B$, subject to $\sigma_A=\sigma_B$. 

Note that in this model, we assume that every player fixes her tactic at the beginning of the game and then follows this tactic. Observe that this does not restrict the players in any way, as throughout the play they receive no information that could influence the choice of the next moves. Indeed, when entering a location, the player only receives the information that the other player is not there, or otherwise the game immediately finishes. Hence, there is no point in considering adaptativity in strategies.

The main result of this work can be now phrased as follows.

\begin{theorem}\label{thm:main}
There exists $\eps>0$ such that for every $n\geq 2$ and every strategy $\sigma\in \Str_n$, we have
    \[\Exp \wtime{\sigma}{\sigma}\geq \frac{n+1}{2}+\eps n.\]
\end{theorem}

In the proof of Theorem~\ref{thm:main} we will use the lower bound for the waiting time for asymmetric strategies of Anderson and Weber~\cite{AndersonW90}. Note that the proof of this result also holds for the game stopped after round $n$.

\begin{theorem}[Anderson and Weber~\cite{AndersonW90}]\label{thm:asymmetric}
For every $n\in \N$ and pair of strategies $\sigma_A,\sigma_B\in \Str_n$, we have
\[\Exp \wtime{\sigma_A}{\sigma_B}\geq \frac{n+1}{2}.\]
\end{theorem}

As mentioned in Section~\ref{sec:intro}, the lower bound provided by Theorem~\ref{thm:asymmetric} is
tight, as witnessed by the wait-for-mommy pair of strategies: $\sigma_A$ is the baby strategy that
deterministically picks a tactic that maps all integers~$i\in [n]$ to $1$, while $\sigma_B$ is the mommy strategy
that deterministically picks the identity function as the tactic.

\section{Proof of \texorpdfstring{Theorem~\ref{thm:main}}{Theorem 1}}

For the rest of the proof we fix the number of locations $n$ to be at least~$2$. For brevity we write
$\Tct\coloneqq\Tct_n$ and $\Str\coloneqq\Str_n$.

\subsection{Passive and active tactics}

Let us start by taking a closer look at the mapping $\sigma_A,\sigma_B\mapsto \Exp
\wtime{\sigma_A}{\sigma_B}$, where $\sigma_A,\sigma_B\in \Str$. We shall try to understand this mapping
from the point of view of linear algebra. 

For tactics $\tau_A,\tau_B\in \Tct$, let
\[\wt{\tau_A}{\tau_B} \coloneqq \Exp\left[\min \left(\{\,i\colon \pi(\tau_A(i))=\tau_B(i)\,\}\cup \{n+1\}\right)\right].\]
Note that here, the tactics $\tau_A,\tau_B$ are fixed and the expectation is taken only over the choice
of the \fit~$\pi$. Let us define a bilinear operator
\[\Phi\colon \R^{\Tct}\times \R^{\Tct}\to
\R\qquad\textrm{as}\qquad \Phi\langle x,y\rangle \coloneqq\sum_{\tau_A,\tau_B\in \Tct}
\wt{\tau_A}{\tau_B}\cdot x_{\tau_A} y_{\tau_B},\]
where $x,y\in \R^{\Tct}$ are vectors indexed by the elements of $\Tct$. Then
\[\Exp \wtime{\sigma_A}{\sigma_B} = \Phi\langle a,b\rangle,\]
where $a,b\in \R^{\Tct}$ are such that $a_\tau$ is the probability of drawing $\tau$ in the distribution
$\sigma_A$, and similarly for~$b_\tau$.

The main idea is as follows. As witnessed by the tightness example for Theorem~\ref{thm:asymmetric}, the
operator $\Phi\langle\cdot,\cdot\rangle$ achieves its minimum possible value when the strategies
$\sigma_A$ and $\sigma_B$ are sort of ``orthogonal''. Namely, one strategy should focus on baby-like
tactics --- being in a few locations and waiting for the other player --- while the other strategy should
focus on mommy-like tactics --- seeking through a large number of location in search of the other player.
Playing a baby-like tactic against a mommy-like tactic yields low waiting time, while the intuition is
that playing two baby-like tactics against each other, or two mommy-like tactics against each other,
should result in waiting time significantly larger than $\frac{n+1}{2}$. When the two players are forced
to use the same strategy, there is a significant probability --- at least $\frac{1}{2}$ --- that they end
up playing tactics of the same kind. This increases the expected waiting time significantly above
$\frac{n+1}{2}$.

We now formalize this intuition, calling baby-like tactics
\emph{passive} and mommy-like tactics \emph{active}.

\begin{definition}
A tactic $\tau\in \Tct$ is called \emph{passive} if $|\tau([n])|\leq n/2$ and \emph{active} otherwise.
The sets of passive and active tactics are denoted by $\Tctf$ and $\Tctm$, respectively.
\end{definition}

In the next sections we will focus on the following lemma.

\begin{lemma}\label{lem:same-sex}
There exists $\delta>0$ such that for all $\tau_A,\tau_B\in \Tct$ satisfying either $\tau_A,\tau_B\in \Tctf$ or $\tau_A,\tau_B\in \Tctm$, we have
\[\wt{\tau_A}{\tau_B}\geq  \frac{n+1}{2}+\delta n.\]
\end{lemma}

Before we proceed to prove Lemma~\ref{lem:same-sex}, let us see how Theorem~\ref{thm:main} follows from
it.

\begin{proof}[of Theorem~\ref{thm:main} assuming Lemma~\ref{lem:same-sex}]
We first note that from Theorem~\ref{thm:asymmetric} applied to two deterministic strategies we may infer that
\begin{equation}\label{eq:raccoon}
\wt{\tau_A}{\tau_B}\geq \frac{n+1}{2}\qquad \textrm{for all }\tau_A,\tau_B\in \Tct.
\end{equation}

Let $a\in \R^\Tct$ be such that $a_\tau$ is the probability that tactic $\tau$ is drawn by the strategy $\sigma$. Write
\[a=a^{\P}+a^{\A},\]
where the supports of $a^{\P}$ and $a^{\A}$ are passive and active tactics, respectively. As $\wt{\cdot}{\cdot}$ is a symmetric function, we have 
\begin{equation}\label{eq:beaver}
\Exp\wtime{\sigma}{\sigma}=\Phi\langle a,a\rangle= \Phi\langle a^{\P},a^{\P}\rangle+\Phi\langle a^{\A},a^{\A}\rangle+2\cdot\Phi\langle a^{\P},a^{\A}\rangle.
\end{equation}
Let $p\coloneqq\sum_{\tau\in \Tctf}a_\tau$ be the probability that $\sigma$ yields a passive tactic. Then, by Lemma~\ref{lem:same-sex}, we have
\begin{eqnarray}
\Phi\langle a^{\P},a^{\P}\rangle & = & \sum_{\tau_A,\tau_B\in \Tctf} \wt{\tau_A}{\tau_B} \cdot a_{\tau_A}a_{\tau_B}\nonumber\\
& \geq & \left(\frac{n+1}{2}+\delta n\right)\cdot \sum_{\tau_A,\tau_B\in \Tctf} a_{\tau_A}a_{\tau_B} =p^2\cdot \left(\frac{n+1}{2}+\delta n\right).\label{eq:female-squirrel}
\end{eqnarray}
Using Lemma~\ref{lem:same-sex} again, we analogously infer that
\begin{equation}\label{eq:male-squirrel}
    \Phi\langle a^{\A},a^{\A}\rangle \geq{(1-p)}^2\cdot \left(\frac{n+1}{2}+\delta n\right).
\end{equation}
A similar computation using~\eqref{eq:raccoon} yields that
\begin{equation}\label{eq:platypus}
\Phi\langle a^{\P},a^{\A}\rangle \geq p(1-p)\cdot \frac{n+1}{2}.
\end{equation}
Finally, letting $\eps\coloneqq\delta/2$ we can combine~\eqref{eq:beaver},~\eqref{eq:female-squirrel},~\eqref{eq:male-squirrel}, and~\eqref{eq:platypus} to conclude that
\begin{eqnarray*}
\Exp\wtime{\sigma}{\sigma} & \geq & (p^2+(1-p)^2+2p(1-p))\cdot \frac{n+1}{2} + p^2\cdot \delta \cdot n + (1-p)^2\cdot \delta \cdot n \\
& = & \frac{n+1}{2} + 2\eps\cdot (p^2+(1-p)^2)\cdot n \geq \frac{n+1}{2} + \eps n,
\end{eqnarray*}
where the last inequality follows from the convexity of the function $x\mapsto x^2$.
\end{proof}

It thus remains to prove Lemma~\ref{lem:same-sex}.

\subsection{High probability of no rendezvous gives high expected waiting time}

We now start analyzing the game when played between a fixed pair of tactics, with the goal of
establishing lower bounds for the expected waiting time till a rendezvous. The intuition is that this
waiting time should be significantly higher than $\frac{n+1}{2}$ provided the probability that during the
$n$ rounds of the game there is no rendezvous at all is bounded from below by some positive constant.
This is made formal in the following lemma.

\begin{lemma}\label{lem:pb-waiting}
Suppose $\tau_A,\tau_B\in \Tct$ are such that
\[\Pb\left(\tau_A(\pi(i))\neq \tau_B(i)\textrm{ for all }i\in [n]\right)\geq \beta\] for some constant $\beta>0$. Then
\[\wt{\tau_A}{\tau_B}\geq \frac{n+1}{2}+\frac{\beta^2}{2}\cdot n.\]
\end{lemma}
\begin{proof}
Let $Z$ be the random variable defined as the waiting time till the first rendezvous, that is, 
\[Z\coloneqq\min\left(\{\,i\colon \pi(\tau_A(i))=\tau_B(i)\,\}\cup \{n+1\}\right).\]
Note that here $\tau_A,\tau_B$ are fixed, so $Z$ depends only on the random choice of the \fit~$\pi$; formally, $Z$ is $\pi$-measurable. Then \[\wt{\tau_B}{\tau_B}=\Exp Z.\]
Observe that $Z$ is a random variable with values in $\{1,2,\ldots,n+1\}$, hence we have
\[\Exp Z = \sum_{k=0}^n \Pb(Z>k).\]
Note that we have $Z>k$ if and only if during the first $k$ rounds the players did not meet. Clearly, during every fixed round, the players meet with probability $\frac{1}{n}$. Hence, by the union bound, the probability that they do not meet during the first $k$ rounds is at least $1-\frac{k}{n}$.
On the other hand, by the assumption of the lemma, this probability is also at least $\beta$. We conclude that
\[\Pb(Z>k)\geq \max\left(1-\frac{k}{n},\beta\right)\qquad\textrm{for all } k\in [n].\]
By combining the above observations it follows that
\begin{eqnarray*}
\wt{\tau_A}{\tau_B} & = & \sum_{k=0}^n \Pb(Z>k) \geq \sum_{k=0}^n \max \left(1-\frac{k}{n},\beta\right) \\
& = & \sum_{k=0}^n \left(1-\frac{k}{n}\right) + \sum_{k=0}^n \max\left(0,\beta-\left(1-\frac{k}{n}\right)\right) \\
& = & \frac{n+1}{2}+\sum_{k=\lceil (1-\beta)n\rceil}^{n} \left(\beta-\left(1-\frac{k}{n}\right)\right)\\
& = & \frac{n+1}{2}+(\beta-1)\cdot (n-\lceil (1-\beta)n\rceil+1)+
\frac{1}{n}\cdot \frac{n+\lceil (1-\beta)n\rceil}{2}\cdot (n-\lceil (1-\beta)n\rceil+1) \\
& = & \frac{n+1}{2}+(n-\lceil (1-\beta)n\rceil+1)\cdot \left(\beta + \frac{\lceil (1-\beta)n\rceil-n}{2n}\right)\\
& \geq & \frac{n+1}{2}+\beta n \cdot \frac{\beta}{2} = \frac{n+1}{2}+\frac{\beta^2}{2}\cdot n.
\end{eqnarray*}
This concludes the proof.
\end{proof}

Thus, by Lemma~\ref{lem:pb-waiting}, for the proof of 
Lemma~\ref{lem:same-sex}
it suffices to show that the probability that no rendezvous occurs throughout the $n$ rounds of the game is bounded away from zero.

\subsection{High variance gives high probability of no rendezvous}\label{sec:var}

Fix a pair of tactics $\tau_A,\tau_B\in \Tct$. Let
\[F\coloneqq\{(\tau_A(i),\tau_B(i))\colon i \in [n]\}\subseteq [n]\times [n].\]
Set $m\coloneqq|F|$ and note that $m\leq n$.
Similarly, for the random \fit~$\pi$, let
\[E(\pi)\coloneqq\{(i,\pi(i))\colon i \in [n]\}\subseteq [n]\times [n].\]
For $f\in F$, let $X_f$ be the indicator random variable taking value $1$ if $f\in E(\pi)$ and $0$ otherwise. Further, let
\[X\coloneqq|F\cap E(\pi)|=\sum_{f\in F} X_f.\]
Note that here $\tau_A,\tau_B$ are considered fixed and $\pi$ is drawn at random, hence~${(X_f)}_{f\in F}$ and therefore~$X$ depend only on the choice of the random \fit~$\pi$; formally, these variables are $\pi$-measurable. Observe that the probability that no rendezvous occurs can be understood in terms of the random variable $X$ as follows:
\begin{equation}\label{eq:crocodile}
\Pb\left(\tau_A(\pi(i))\neq \tau_B(i)\textrm{ for all }i\in [n]\right) = \Pb(X=0).
\end{equation}
From now on, we adopt the above notation whenever the pair of tactics $\tau_A,\tau_B$ is clear from the context.

The next lemma is the key conceptual step in the proof. We show that in order to give a lower bound on the probability that no rendezvous occurs, it suffices to give a lower bound on the variance of $X$.

\begin{lemma}\label{lem:var-pb}
Suppose $\tau_A,\tau_B\in \Tct$ are such that
\[m\geq \left(1-\sqrt{\alpha/2}\right)\cdot n\qquad\textrm{and}\qquad\Var X\geq \alpha,\] for some constant $\alpha>0$. Then
\[\Pb(X=0)\geq \frac{\alpha^2}{128}.\]
\end{lemma}

The proof of Lemma~\ref{lem:var-pb} spans the rest of this section.
The intuition is that high variance of $X$ means that $X$ is not well concentrated around its mean, which in turns implies that the probability of it being below the mean --- equivalently equal to $0$ --- is high. Hence, we need to understand the mean of $X$ as well as estimate its higher moments.

Observe that if~$f=(i,j)\in F$, then the probability that $\pi(i)=j$ is equal to $\frac{1}{n}$. Hence, $X_f$ takes value~$1$ with probability $\frac{1}{n}$ and $0$ with probability $1-\frac{1}{n}$. Consequently, we have
\[\Exp X_f=\frac{1}{n}\qquad \textrm{for each }f\in F.\]
By linearity of expectation,
\[\Exp X = \frac{m}{n}\leq 1.\]
In the sequel we will also need an upper bound on the fourth central moment of~$X$, that is, on~$\Exp|X-\Exp X|^4$.
To this end, we first establish, in the next two assertions, an upper bound on the fourth moment of $X$, that is, on~$\Exp X^4$.

\begin{claim}\label{cl:pairs}
For pairwise different pairs $e,f,g,h\in F$, we have
\begin{eqnarray*}
\Exp X_eX_f & \leq & \frac{1}{n(n-1)}\\
\Exp X_eX_fX_g & \leq & \frac{1}{n(n-1)(n-2)}\\
\Exp X_eX_fX_gX_h & \leq & \frac{1}{n(n-1)(n-2)(n-3)}.
\end{eqnarray*}
\end{claim}
\begin{proof}
Let us focus on the first inequality. Write $e=(i,j)$ and $f=(i',j')$. 
Observe that if $i=i'$ or~$j=j'$, then $X_e$ and $X_f$ cannot simultaneously be equal to $1$ since~$e\neq f$, and hence $X_eX_f=0$ surely. Otherwise, the probability that for $\pi$ chosen uniformly at random we have $\pi(i)=j$ and $\pi(i')=j'$ is $\frac{1}{n(n-1)}$. Consequently $\Pb(X_eX_f=1)=\frac{1}{n(n-1)}$. This implies the first inequality. The proofs of the remaining two inequalities are analogous.
\cqed\end{proof}

\begin{claim}\label{cl:fourth}
It holds that
\[\Exp X^4\leq 15.\]
\end{claim}
\begin{proof}
    For each $e\in F$, since~$X_e\in\{0,1\}$ we have $X_e=X_e^2=X_e^3=X_e^4$.
By Assertion~\ref{cl:pairs} and the fact that $m\leq n$, we have
\begin{eqnarray*}
\Exp X^4 & = & \Exp \left(\sum_{e\in F} X_e\right)^4 \\
         & = & \sum_{e\in F} \Exp X_e^4 + \sum_{\{e,f\}\subseteq F} \Exp (4X_e^3X_f+6X_e^2X_f^2+4X_eX_f^3) \\
         & & + \sum_{\{e,f,g\}\subseteq F} \Exp (12X_e^2X_fX_g + 12X_eX_f^2X_g+12X_eX_fX_g^2)\\
         & & + \sum_{\{e,f,g,h\}\subseteq F} \Exp (24X_eX_fX_gX_h)\\
         & = & \sum_{e\in F} \Exp X_e + 14\sum_{\{e,f\}\subseteq F} \Exp X_eX_f \\
         & & + 36\sum_{\{e,f,g\}\subseteq F} \Exp X_eX_fX_g + 24 \sum_{\{e,f,g,h\}\subseteq F} \Exp X_eX_fX_gX_h\\
         & \leq & \frac{m}{n} + 14\cdot \frac{\binom{m}{2}}{n(n-1)} + 36\cdot \frac{\binom{m}{3}}{n(n-1)(n-2)}+24\cdot  \frac{\binom{m}{4}}{n(n-1)(n-2)(n-3)}\\
         & \leq & 1+7+6+1=15.
\end{eqnarray*}
This concludes the proof.
\cqed\end{proof}

We will also use the following well-known anti-concentration inequality.

\begin{theorem}[Paley-Zygmund inequality,~\cite{PaleyZ32}]\label{thm:pz}
Let $Z$ be a non-negative random variable with finite variance and let $\lambda\in[0,1]$. Then
\[\Pb(Z\geq \lambda\Exp Z)\geq (1-\lambda)^2\cdot \frac{(\Exp Z)^2}{\Exp Z^2}.\]
\end{theorem}

With all the tools prepared, we proceed with the proof of Lemma~\ref{lem:var-pb}.
We use Theorem~\ref{thm:pz} with $\lambda=\frac{1}{2}$ for the random variable
\[Z\coloneqq|X-\Exp X|^2.\]
By Assertion~\ref{cl:fourth} and the fact that $\Exp X\leq 1$, we have
\[\Exp Z^2 = \Exp |X-\Exp X|^4\leq 1+\Exp X^4\leq 16.\]
As $\Exp Z=\Var X\geq \alpha$, from Theorem~\ref{thm:pz} we infer that
\begin{equation}\label{eq:alligator}
\Pb\left(Z\geq \alpha/2\right)\geq
\Pb\left(Z\geq \Exp Z/2\right)\geq 
\frac{1}{4}\cdot \frac{(\Exp Z)^2}{16}\geq 
\frac{1}{4}\cdot \frac{\alpha^2}{16}=\frac{\alpha^2}{64}.
\end{equation}
Observe now that the assumption that $m>\left(1-\sqrt{\alpha/2}\right)\cdot n$ implies that 
\[1-\Exp X =1-\frac{m}{n}<\sqrt{\alpha/2}.\]
This, in turns, implies that the event
\[\left\{|X-\Exp X|^2\geq \alpha/2\right\}\]
is disjoint with the event $\{X=1\}$. By combining this with~\eqref{eq:alligator}, we conclude that
\[\Pb(X\neq 1)\geq \Pb\left(|X-\Exp X|^2\geq \alpha/2\right)=
\Pb\left(Z\geq \alpha/2\right)\geq \frac{\alpha^2}{64}.\]
Since $X$ is a non-negative integer-valued random variable with mean not larger than~$1$, we have 
\[\Pb(X\neq 1)=\Pb(X=0)+\Pb(X\geq 2)\qquad\textrm{and}\qquad \Pb(X=0)\geq \Pb(X\geq 2).\]
By combining the two inequalities above we conclude that
\[\Pb(X=0)\geq \frac{1}{2}\cdot \Pb(X\neq 1)\geq \frac{\alpha^2}{128}.\]
This concludes the proof of Lemma~\ref{lem:var-pb}.

\subsection{Many disjoint pairs give high variance}

Two pairs $(i,j)$ and~$(i',j')$, each in $[n]\times[n]$, are \emph{disjoint} if $i\neq i'$ and $j\neq j'$. We now prove that to ensure that for a pair of tactics $\tau_A,\tau_B$, the variance of $X$ is high, it suffices to show that among pairs in $F$, there is a quadratic number of pairs of pairs that are disjoint.

\begin{lemma}\label{lem:dp-var}
Suppose $\tau_A,\tau_B\in \Tct$ are such that there are at least $\alpha \binom{n}{2}$ disjoint pairs in~$F$, for some positive constant~$\alpha$. Then $\Var X\geq \alpha$.
\end{lemma}
\begin{proof}
As in the proof of Assertion~\ref{cl:pairs}, we observe that for every pair of different elements $e,f\in F$, we have
\[\Exp X_eX_f=\begin{cases}\frac{1}{n(n-1)} & \text{if $e$ and $f$ are disjoint;}\\
0 & \text{otherwise.}\end{cases}
    \]
Therefore, for all different $e,f\in F$ we have
\begin{eqnarray*}
\Var X_e & = & \Exp X_e^2 - (\Exp X_e)^2 = \frac{n-1}{n^2},\quad\text{and}\\
\Cov (X_e,X_f) & = & \Exp X_eX_f - \Exp X_e\Exp X_f = [e\cap f=\emptyset]\cdot\frac{1}{n(n-1)}-\frac{1}{n^2},
\end{eqnarray*}
where the expression $[e\cap f=\emptyset]$ takes value $1$ if $e$ and $f$ are disjoint, and $0$ otherwise. Consequently,
\begin{eqnarray*}
\Var X & = & \sum_{e\in F} \Var X_e + 2\cdot \sum_{\{e,f\}\subseteq F}\Cov(X_e,X_f)\\
 & \geq & m\cdot \frac{n-1}{n^2}-2\cdot \binom{m}{2}\cdot \frac{1}{n^2} + 2\cdot \alpha\binom{n}{2}\cdot \frac{1}{n(n-1)} \\
 & = & \frac{m(n-1)-m(m-1)}{n^2}+\alpha,
\end{eqnarray*}
which is at least~$\alpha$ because~$m\leq n$.  This concludes the proof.
\end{proof}

\subsection{Finding many disjoint pairs}

Finally, we prove that if $\tau_A$ and $\tau_B$ are two tactics of the same kind, then the set of pairs
$F$ defined for $\tau_A$ and $\tau_B$ contains many pairs of disjoint pairs. For this, it will be
convenient to interpret $F$ as the edge set of a bipartite graph, with each side of the bipartition
consisting of a copy of the set $[n]$. In this view, a pair of disjoint pairs corresponds to a pair of
disjoint edges: two edges in a graph being \emph{disjoint} if all the four endpoints of these edges are
pairwise different. 

We first prove the following graph-theoretic lemma. The degree~$\deg(u)$ of a vertex~$u$ in a graph~$G$
is the number of edges of~$G$ incident to~$u$.

\begin{lemma}\label{lem:many-disjoint}
Let $G=(A,B,E)$ be a bipartite graph such that $A$ and $B$ --- the sides of the bipartition --- have size
    $n$ each, $\frac{11}{12}n\leq |E|\leq n$, and the degree of each vertex in $G$ is at most
    $\frac{2}{3}n$.  Then there are two disjoint subsets of edges $E_1,E_2\subseteq E$, each of size at
    least $n/8$, such that every edge from $E_1$ is disjoint with every edge in $E_2$.
\end{lemma}
\begin{proof}
For $X\subseteq A\cup B$, we let~$\deg(X)\coloneqq\sum_{u\in X}\deg(u)$.

Let $a_1,\ldots,a_n$ be the vertices of $A$ in non-increasing order with respect to their degrees. Let
    $t\in \{0,1,\ldots,n\}$ be the largest index such that $A_1\coloneqq\{a_1,\ldots,a_t\}$ satisfies
    $\deg(A_1)\leq \frac{2}{3}n$. Since the degree of every vertex is at most $\frac{2}{3}n$ and
    $|E|>\frac{2}{3}n$, we know that neither~$A_1$ nor~$A_2\coloneqq A\setminus A_1$ is empty. In other
    words,~$t\in\{1,\dotsc,n-1\}$.  Further, since $\deg(A_1)\leq \frac{2}{3}n$, $\deg(A_1\cup
    \{a_{t+1}\})>\frac{2}{3}n$, and $\deg(a_{t+1})\leq \deg(v)$ for every $v\in A_1$, it follows that
    $\deg(A_1)> n/3$. Since $\deg(A_1)\leq \frac{2}{3}n$, and $\deg(A)=|E|\geq \frac{11}{12}n$, we also
    have $\deg(A_2)\geq n/4$. We conclude that we have found a partition $A_1\uplus A_2$ of $A$ such that
\[\deg(A_1)\geq n/4\qquad \textrm{and}\qquad \deg(A_2)\geq n/4.\]
    Symmetrically, we can find a partition $B_1\uplus B_2$ of $B$ such that
\[\deg(B_1)\geq n/4\qquad \textrm{and}\qquad \deg(B_2)\geq n/4.\]

For all~$s,t\in\{1,2\}$, let $F_{st}$ be the set of all edges from $E$ with one endpoint in $A_s$ and the
    other in $B_t$, and set~$m_{st}\coloneqq|F_{st}|$.
    The above lower bounds on the degrees of $A_1,A_2,B_1,B_2$ imply that
\begin{equation}\label{eq:octopus}
m_{11}+m_{12}\geq n/4,\qquad m_{21}+m_{22}\geq n/4,\qquad m_{11}+m_{21}\geq n/4,\qquad m_{12}+m_{22}\geq n/4.
\end{equation}
Observe that if $m_{11}\geq n/8$ and $m_{22}\geq n/8$, then $E_1=F_{11}$ and $E_2=F_{22}$ satisfy the
    condition from the lemma statement.  Similarly, if $m_{12}\geq n/8$ and $m_{21}\geq n/8$, then taking
    $E_1=F_{12}$ and $E_2=F_{21}$ concludes the proof.  We are thus left with the case when there is
    $st\in \{11,22\}$ such that $m_{st}<n/8$ and there is $s't'\in \{12,21\}$ such that $m_{s't'}<n/8$.
    But then~$m_{st}+m_{s't'}<n/4$, which contradicts one of the inequalities~\eqref{eq:octopus}.
\end{proof}

From Lemma~\ref{lem:many-disjoint} we immediately infer the following result.

\begin{lemma}\label{lem:same-sex-variance}
Suppose that $\tau_A,\tau_B\in \Tct$ is a pair of tactics such that $\tau_A,\tau_B\in \Tctf$ or $\tau_A,\tau_B\in \Tctm$, and that $|F|\geq \frac{11}{12}n$. 
Then $\Var X\geq \frac{1}{32}$.
\end{lemma}
\begin{proof}
Let $G=(A,B,F)$ be the bipartite graph constructed by taking $A$ and $B$ to be two disjoint copies of the set $[n]$, and interpreting each pair $(i,j)\in F$
as an edge that connects the copy of $i$ in~$A$ with the copy of~$j$ in~$B$.
We now verify that $G$ satisfies the prerequisites of Lemma~\ref{lem:many-disjoint}. We have $|F|\geq \frac{11}{12}n$ by assumption, so we are left with checking the requirements on degrees.

Suppose first $\tau_A,\tau_B\in \Tctf$. Then $|\tau_A([n])|\leq n/2$, so there are only at most $n/2$ indices $i\in [n]$ that may be the first coordinates of pairs from~$F$. Hence in $G$, the degree of every vertex in $B$ is at most $n/2$. A symmetric reasoning shows that the degree of every vertex in $A$ is at most $n/2$.

Suppose now that $\tau_A,\tau_B\in \Tctm$. Then $|\tau_A([n])|>n/2$, hence there are at least $\frac{n+1}{2}$ indices $i\in [n]$ that are the first coordinates of pairs from~$F$.
Every $i\in [n]$ is the first coordinate of at most $\frac{n+1}{2}$ pairs from~$F$. Indeed, otherwise it
    would not be possible that each of the at least $\frac{n-1}{2}$ indices $i'\in \tau_A([n])\setminus
    \{i\}$ would be the first coordinate of one of the remaining less than $\frac{n-1}{2}$ pairs from
    $F$. This means that in $G$, the degree of each vertex from $A$ is at most~$\frac{n+1}{2}\leq
    \frac{2}{3}n$. A symmetric reasoning shows that the degree of each vertex from $B$ is at
    most~$\frac{2}{3}n$.

Having verified the prerequisites of Lemma~\ref{lem:many-disjoint}, we can conclude that there exist disjoint subsets of pairs~$F_1,F_2\subseteq F$, each of size at least $n/8$, such that every pair from $F_1$ is disjoint with every pair from~$F_2$. This implies that in $F$ there are at least $\frac{n^2}{64}\geq \frac{1}{32}\cdot \binom{n}{2}$ pairs of pairs that are disjoint. By Lemma~\ref{lem:dp-var}, this implies that $\Var X\geq \frac{1}{32}$.
\end{proof}

\subsection{Wrapping up the proof}

With all the tools prepared, we are now in a position to prove Lemma~\ref{lem:same-sex}.

\begin{proof}[of Lemma~\ref{lem:same-sex}]
Let $F$, $m$, and $X$ be defined for $\tau_A,\tau_B$ as in Section~\ref{sec:var}.

We first consider the corner case when $m\leq \frac{11}{12}n$. Then
\[\Exp X = \frac{m}{n}\leq \frac{11}{12}.\]
Therefore, by Markov's inequality we infer that
\[\Pb(X=0)=1-\Pb(X\geq 1)\geq 1-\frac{11}{12}=\frac{1}{12}.\]

Now consider the case when $m>\frac{11}{12}n$. 
By Lemma~\ref{lem:same-sex-variance} we infer that $\Var X\geq \frac{1}{32}$.
Applying Lemma~\ref{lem:var-pb} for $\alpha=\frac{1}{32}$, we conclude that in this case
\[\Pb(X=0)\geq \frac{1}{128\cdot 32^2}=2^{-17}.\]
Note here that the assumption $m\geq \left(1-\sqrt{\alpha/2}\right)\cdot n$ is satisfied, because $1-\sqrt{\alpha/2}=\frac{7}{8}<\frac{11}{12}$.

Hence, we have $\Pb(X=0)\geq 2^{-17}$ in both cases. By Lemma~\ref{lem:pb-waiting} we now conclude that
\[\wt{\tau_A}{\tau_B}\geq \frac{n+1}{2}+2^{-35}\cdot n.\]
Hence, Lemma~\ref{lem:same-sex} holds for $\delta=2^{-35}$.
\end{proof}

Recalling that the proof of Theorem~\ref{thm:main} sets~$\eps$ to be $\delta/2$, we conclude that Theorem~\ref{thm:main} holds for~$\eps=2^{-36}$.

\bigskip
\noindent
\textbf{Acknowledgements.} The last author thanks Amos Korman for stimulating discussions on the
rendezvous problem in general.


\begin{bibdiv} 
\begin{biblist} 
\bib{Alpern95}{article}{
   author={Alpern, Steve},
   title={The rendezvous search problem},
   journal={SIAM J. Control Optim.},
   volume={33},
   date={1995},
   number={3},
   pages={673--683},
   issn={0363-0129},
   review={\MR{1327232}},
   doi={10.1137/S0363012993249195},
}

\bib{Alpern02}{article}{
   author={Alpern, Steve},
   title={Rendezvous search: a personal perspective},
   journal={Oper. Res.},
   volume={50},
   date={2002},
   number={5},
   pages={772--795},
   issn={0030-364X},
   review={\MR{1923700}},
   doi={10.1287/opre.50.5.772.363},
}

\bib{Alpern02b}{article}{
   author={Alpern, Steve},
   title={Rendezvous search on labeled networks},
   journal={Naval Res. Logist.},
   volume={49},
   date={2002},
   number={3},
   pages={256--274},
   issn={0894-069X},
   review={\MR{1885643}},
   doi={10.1002/nav.10011},
}

\bib{Alpern13}{article}{
   author={Alpern, Steve},
   title={Ten open problems in rendezvous search},
   conference={
      title={Search theory},
   },
   book={
      publisher={Springer, New York},
   },
   date={2013},
   pages={223--230},
   review={\MR{3087871}},
   doi={10.1007/978-1-4614-6825-7\_14},
}

\bib{AlpernL02}{article}{
   author={Alpern, Steve},
   author={Lim, Wei Shi},
   title={Rendezvous of three agents on the line},
   journal={Naval Res. Logist.},
   volume={49},
   date={2002},
   number={3},
   pages={244--255},
   issn={0894-069X},
   review={\MR{1885642}},
   doi={10.1002/nav.10005},
}

\bib{AndersonW90}{article}{
   author={Anderson, E. J.},
   author={Weber, R. R.},
   title={The rendezvous problem on discrete locations},
   journal={J. Appl. Probab.},
   volume={27},
   date={1990},
   number={4},
   pages={839--851},
   issn={0021-9002},
   review={\MR{1077533}},
   doi={10.2307/3214827},
}

\bib{PaleyZ32}{article}{
    author = {Paley, R. E. A. C},
    author = {Zygmund, A.}
    title = {On some series of functions, (3)},
    journal = {Math. Proc. Cambridge Philos. Soc.}, 
    volume = {28},
    number = {2},
    pages = {190--205},
    year = {1932},
    doi = {10.1017/S0305004100010860},
}

\bib{Weber09}{report}{
    title={The {A}nderson-{W}eber strategy is not optimal for symmetric rendezvous search on {$K_4$}},
    author={Weber, Richard},
    eprint={0912.0670},
    url = {https://arxiv.org/abs/0912.0670},
    status = {unpublished},
    year = {2009},
} 

\bib{Weber12}{article}{
   author={Weber, Richard},
   title={Optimal symmetric Rendezvous search on three locations},
   journal={Math. Oper. Res.},
   volume={37},
   date={2012},
   number={1},
   pages={111--122},
   issn={0364-765X},
   review={\MR{2891149}},
   doi={10.1287/moor.1110.0528},
}

\bib{Weber12b}{article}{
   author={Weber, Richard},
   title={Strategy for quickest second meeting of two agents in two
   locations},
   journal={Math. Oper. Res.},
   volume={37},
   date={2012},
   number={1},
   pages={123--128},
   issn={0364-765X},
   review={\MR{2891150}},
   doi={10.1287/moor.1110.0529},
}
\end{biblist} 
\end{bibdiv}
\end{document}